%% file: fairness_netecon.tex
\documentclass[envcountsame,runningheads,11pt]{llncs}
\usepackage[margin=1in]{geometry}
\usepackage{setspace}
\usepackage{amssymb,amsmath}
\usepackage{graphicx}
\usepackage{url}
\usepackage{paralist}
\usepackage{algorithm}
\usepackage{environ}
\usepackage{cite}
\usepackage{tikz}
\usetikzlibrary{shapes,calc,intersections}

\urldef{\mailsc}\path | {ashishg,rmhulett,anilesh}@stanford.edu |

\newtheorem{result}{Result}
\newcommand{\cc}{\mathcal{C}}
\newcommand{\V}{\mathcal{V}}

\newcommand{\s}{\mathcal{S}}

\newcommand{\vpr}{v^\prime}

\newcommand{\Rplus}{\mathcal{R}_{\geq 0}}

\newcommand{\cpr}{c^\prime}

\usepackage{bbm}
\usepackage{color}

\newcommand{\dpr}{d^\prime}

\newcommand{\dist}{\mathrm{distortion}}
\newcommand{\fratio}{\mathrm{fairness}}

\newcommand{\opt}{c_\mathrm{opt}}
\newcommand{\alg}{c_\mathrm{alg}}
\newcommand{\argmin}{\mathrm{argmin}}
\newcommand{\argmax}{\mathrm{argmax}}

\begin{document}

\mainmatter  
\singlespacing
\title{Relating Metric Distortion and Fairness of Social Choice Rules}


\author{Ashish Goel, Reyna Hulett, Anilesh K. Krishnaswamy\\
\mailsc}
\authorrunning{Goel, Hulett and Krishnaswamy}

\institute{
Stanford University \\
}
\maketitle

\begin{abstract}
One way of evaluating social choice (voting) rules is through a utilitarian distortion framework. In this model, we assume that agents submit full rankings over the alternatives, and these rankings are generated from underlying, but unknown, quantitative costs. The \emph{distortion} of a social choice rule is then the ratio of the total social cost of the chosen alternative to the optimal social cost of any alternative; since the true costs are unknown, we consider the worst-case distortion over all possible underlying costs. Analogously, we can consider the worst-case \emph{fairness ratio} of a social choice rule by comparing a useful notion of fairness (based on approximate majorization) for the chosen alternative to that of the optimal alternative. With an additional metric assumption -- that the costs equal the agent-alternative distances in some metric space -- it is known that the Copeland rule achieves both a distortion and fairness ratio of at most 5. For other rules, only bounds on the distortion are known, e.g., the popular Single Transferable Vote (STV) rule has distortion $O(\log m)$, where $m$ is the number of alternatives. We prove that the distinct notions of distortion and fairness ratio are in fact closely linked -- within an additive factor of 2 for any voting rule -- and thus STV also achieves an $O(\log m)$ fairness ratio. We further extend the notions of distortion and fairness ratio to social choice rules choosing a \emph{set} of alternatives. By relating the distortion of single-winner rules to multiple-winner rules, we establish that Recursive Copeland achieves a distortion of 5 and a fairness ratio of at most 7 for choosing a set of alternatives.
\end{abstract}


\section{Introduction}\label{sec:intro}
Social choice theory studies the aggregation of agent preferences into a single collective decision via a social choice rule. Often these preferences are expressed as total orderings over a set of possible alternatives, and a social choice rule maps any instance of preferences to one or more alternatives. The traditional approach to evaluating the quality of social choice rules has been a normative, axiomatic one. A great deal of work has been done to propose various axioms or properties that social choice rules ought to satisfy -- such as strategy-proofness or the majority winner criterion -- and evaluate different rules by which criteria they meet \cite{arrow2012social}. Unfortunately, even small sets of natural axioms may be impossible to satisfy simultaneously. For instance, the celebrated Gibbard-Satterthwaite Theorem rules out the existence of strategy-proof, non-dictatorial social choice rules with more than two alternatives \cite{gibbard1973manipulation,satterthwaite1975strategy}. Thus we must either accept social choice rules which fail to satisfy some natural properties, or make assumptions which limit the possible agent preferences but permit strategy-proof rules \cite{barbera2001introduction,moulin1980strategy}.

An alternative to the axiomatic approach, which has lately received much attention in the field of computational social choice \cite{boutilier2015optimal,caragiannis2011voting,procaccia2006distortion}, is to adopt a utilitarian view -- agents express their ordinal preferences by ranking the alternatives, but have latent \emph{cardinal} preferences over the alternatives. In particular, much work \cite{anshelevich2015approximating,anshelevich2016randomized,skowron2017social} has been done on the \emph{metric distortion} problem \cite{anshelevich2015approximating}. Under this model, agents and alternatives are assumed to lie in an unknown, arbitrary metric space, and an agent's cost for an alternative is given by the distance between the two. Social choice rules are viewed as approximation algorithms trying to choose the alternative with the lowest social cost, given access to only the ordinal preferences of agents. Similar to the competitive ratio of online approximation algorithms, the quantity of interest here is the worst-case value (over all possible underlying costs) of the \emph{distortion} -- the ratio of the social cost of the chosen alternative to that of the optimal alternative, chosen omnisciently \cite{anshelevich2015approximating}.
In this setting, the best known positive result for deterministic social rules is that the distortion of the Copeland rule is at most $5$ \cite{anshelevich2015approximating}. Another recent result establishes that the distortion of the Single Transferable Vote (STV) rule, which is widely used in practice, is $O(\log m)$ \cite{skowron2017social}, where $m$ is the number of alternatives. 

Under the metric distortion approach, another important question is to quantify how ``fair'' choosing a particular alternative is. Since the costs incurred could vary widely among the agents, minimizing the total cost might not be the only useful objective. For instance, imagine there are two agents and two alternatives, and the costs for the agents are $\{1,3\}$ for the first alternative, and $\{2,2\}$ for the second. It is natural to expect that under any reasonable notion of fairness the second is more desirable than the first. Although many notions of fairness exist, not all are applicable in this setting, and some cannot even be bounded for any social choice rule \cite{goel2017metric,anshelevich2015approximating}. We use the \emph{fairness ratio} \cite{goel2017metric}, which is a simultaneous bound, over all $k$, on the distortion of the social cost objective given by the sum of the $k$ largest agent costs. This definition of fairness is based on the notion of approximate majorization \cite{goel2006simultaneous}, which generalizes the various notions of fairness that have been studied in the context of routing, bandwidth allocation and load balancing problems \cite{kleinberg1999fairness,kumar2000fairness,goel2001approximate}. For instance, it includes as special cases both total cost minimization (utilitarianism), when $k=N$ (where N is the total number of agents), and max-min fairness (egalitarianism), when $k=1$. A bound on the fairness ratio also translates to an approximation guarantee on a wide class of convex objective functions \cite{goel2017metric}. For example, a constant factor bound on the fairness ratio implies the same bound on the distortion of social cost objectives such as the geometric mean, or any $l_p$ norm of the agent costs \cite{goel2017metric}. Thus, this interesting notion of fairness captures a wide range of desirable properties.

It is known that Copeland achieves a fairness ratio of at most 5, thereby yielding a constant-factor approximation for a large class of convex objectives \cite{goel2017metric}. Since it also achieves a distortion of at most 5, a natural question arises as to whether the distortion and fairness ratio are inherently connected. For many popular voting rules -- such as scoring rules and STV -- bounds are known on the distortion but not the fairness ratio. Could the distortion bounds for these rules extend to the fairness ratio as well, as is the case with Copeland? In particular, does any bound on distortion directly imply a bound on the fairness ratio? The primary aim of this paper is to answer this question.

The majority of previous work on the metric distortion problem looks only at social choice rules that choose a single alternative and not those that choose a set of alternatives. To address this lacuna, we also study the problem of establishing bounds on the distortion and fairness ratio for the case of choosing a set of winners, given the size of the desired set. Applications of multi-winner elections are quite diverse such as proportional representation in parliament \cite{monroe1995fully,betzler2013computation}, selecting a diverse committee \cite{chamberlin1983representative} (e.g. locations of fire stations), offering a selection of movies to passengers on a flight \cite{lu2011budgeted}, and shortlisting candidates for a job interview \cite{barbera2008choose,elkind2017properties}. We will focus on settings where a set of resources has to be chosen to serve a community of voters. Some examples of these setting are choosing overlay networks in the Internet \cite{andersen2001resilient}, or a set of public projects to implement given a fixed budget \cite{cabannes2004participatory,goel2015knapsack}. We will also assume in such settings that the preferences are additive -- the total cost for an agent is just the sum of her costs over the chosen set of alternatives. It is fairly straightforward to see that a repeated application of Copeland to choose a set of winners yields a bound of 5 on the distortion. In fact, we establish that any single-winner rule can be applied recursively to obtain a multiple-winner rule with the same bound on distortion. However, it was not known if a constant-factor bound on the fairness ratio is possible in this case. We answer this question in the positive by extending the distortion-fairness relationship to rules choosing a set of winners, thus establishing that the set chosen by a repeated application of Copeland achieves a bound of at most 7 on the fairness ratio.

\subsection{Our Results}

Our primary focus in this paper will be to quantify the relationship between the distortion and fairness ratio for social choice rules that choose a single winner, which also leads to new results for any rules where bounds are known on the distortion but not the fairness ratio. We then extend these results to rules that choose a set of winners, and characterize upper bounds on distortion and fairness in this setting. We also provide a lower bound on the gap between the distortion and fairness ratio.

\emph{Note: All our proofs are provided separately in the appendix.}

\subsubsection{Distortion and Fairness for Single-Winner Rules}

Given the ordinal preferences of the agents, the distortion of an alternative is the worst case value (over all possible underlying metrics) of the ratio between the sum of agent costs for it, and the sum of the agent costs for the optimal alternative.  For studying fairness under the metric distortion framework, we instead consider the ratio of the sum of the $k$ largest agent costs for an alternative to the $k$ largest costs for the optimal alternative, and take the maximum value among these ratios over all possible values of $k$. The overall fairness ratio is the worst case value of this quantity over all possible underlying metrics. Since for $k=N$ we get the sum of costs objective, the fairness ratio is at least as large as distortion.

Our main result is that the distortion and fairness ratio are closely related. For any given ordinal preferences, and any alternative, we show that the fairness ratio cannot be much larger than the distortion. 

\begin{result}[Theorem~\ref{thm:2}]
Given the ordinal preferences of the agents, the fairness ratio of any alternative is at most 2 more than its distortion.
\end{result}

We also present instances where the ``obvious'' winner for any reasonable social choice rule (which is also the distortion minimizing alternative) has a distortion-fairness gap approaching 2.

\begin{result}[Theorem~\ref{thm:tight}]
There exist instances where the gap between the distortion and fairness ratio of the distortion-optimal alternative approaches 2.
\end{result}

In fact, because the above results hold for any set of preferences and any metric space, they also imply that for any social choice rule, the worst-case fairness ratio is at most 2 more than its worst-case distortion. This implies novel bounds on the fairness ratio of several social choice rules, including STV, for which only bounds on distortion were previously known \cite{skowron2017social}.

\begin{result}[Corollary~\ref{cor:rules}]
The fairness ratio of STV is $O(\log m)$ and $\Omega (\sqrt{\log m})$.
\end{result}

\subsubsection{Distortion and Fairness for Multiple-Winner Rules}
We also consider the case where a social choice rule selects a set of winners (of a given size). Here, the cost of an agent for any set of alternatives is the sum of her costs over the alternatives in the set, and the social cost is the sum of the individual agent costs, as before.

It is straightforward to see from previous results \cite{anshelevich2015approximating} that choosing the desired set by a repeated application of the Copeland rule achieves a worst-case distortion of 5. In fact, we show that recursive application of any single-winner rule -- where a set of given size $\ell$ is selected by first choosing the winner when the rule is applied over the entire set of $m$ alternatives, then over the remaining $m-1$ alternatives, and so on until we have chosen $\ell$ alternatives -- results in an $\ell$-winner rule with the same distortion bound.

\begin{result}[Theorem~\ref{thm:distsets}]
Recursive application of a single-winner social choice rule gives an $\ell$-winner rule (for any $\ell$) achieving the same bound on distortion.
\end{result}

Although this result does not directly extend to the fairness ratio, we do establish that our main result relating distortion and fairness \emph{does} apply to multiple-winner rules. Combining this and the above result concerning recursive social choice rules, we show that it is possible to achieve, using Recursive Copeland, a constant-factor bound on the multiple-winner fairness ratio.

\begin{result}[Corollary~\ref{cor:recursive}]
Recursive Copeland achieves a worst-case distortion of at most 5, and a worst-case fairness ratio of at most 7.
\end{result}

However, note that we do not know of a better lower bound than 5 for the worst-case fairness ratio of Recursive Copeland. The lower bound of 5 follows from previous work -- the upper bound of 5 for both the distortion and fairness ratio of Copeland is known to be tight \cite{anshelevich2015approximating,goel2017metric}. 

\subsection{Related Literature}
Within the general model of agents with cardinal preferences who report only ordinal information, several ways of bounding distortion have been studied in the literature \cite{moulin2016handbook}. It is usually assumed that agents' ordinal rankings straightforwardly match the order of their cardinal utilities, with the most-preferred alternative ranked first. In this case, it is known that when the underlying utilities are unrestricted, or simply normalized, the worst-case distortion of any deterministic social choice rule is large \cite{procaccia2006distortion}. With randomized mechanisms, it is possible to achieve an expected distortion of $O(\sqrt{m} \log^* m)$, where $m$ is the number of alternatives \cite{boutilier2015optimal}. However, if the mechanism has control over how agents' utilities are translated into rankings, it is possible to design randomized rules with very low distortion \cite{caragiannis2011voting}. The preceding results are not restricted to strategy-proof voting rules, but it is possible to construct a truthful-in-expectation mechanism whose worst-case distortion is $O(m^{3/4})$ \cite{filos2014truthful}.

One way to improve these distortion results is to make spatial assumptions on the underlying cardinal preferences, a technique which has a long history in social choice \cite{enelow1984spatial,moulin1980strategy}. Such models, especially those using Euclidean spaces, have naturally also been studied in the approximation algorithms literature on facility location problems \cite{arya2004local,drezner1995facility}. In these models, the cost of an agent for an alternative is given by the distance between the two.

As mentioned earlier, our work follows the literature on the analysis of distortion of social choice rules under the assumption that agent costs form an unknown metric space \cite{anshelevich2015approximating,anshelevich2016randomized}. Several lower and upper bounds for the sum of costs and median objectives are known, in both deterministic and randomized settings \cite{anshelevich2015approximating,anshelevich2016randomized}. For example, in the deterministic case, Copeland achieves a distortion of at most 5 \cite{anshelevich2015approximating}, and the distortion of STV is $O(\log m)$ \cite{skowron2017social}. In the special case of Euclidean metrics, it possible to design low-distortion mechanisms, with the additional constraint of their being truthful-in-expectation \cite{feldman2016voting}.

Social choice rules choosing sets of alternatives rather than a single winner have been studied for some time \cite{chamberlin1983representative,monroe1995fully,faliszewski2017multiwinner}, but recently they have also been evaluated within the distortion framework \cite{caragiannis2017subset}. In most of these settings, the value of a set for an agent is determined by her favorite item in the set. Such rules have been evaluated both using distortion of utilities and an additive notion of approximation, regret; the distortion in this setting remains large -- $\Theta(\sqrt{m})$ -- at least when the number of winners is small compared to the number of alternatives \cite{caragiannis2017subset}. To the best of our knowledge, multiple-winner rules have not been studied under the metric costs model.
 
In the distortion framework, both the interpersonal comparison of utilities, and the goal of utility maximization, are implicitly assumed to be valid. While the interpersonal comparison of utilities is more meaningful in some contexts than others \cite{boutilier2015optimal}, we take it for granted. 
 
In addition to distortion, we borrow from the various notions of fairness studied in the context of network problems such as bandwidth allocation and load balancing \cite{kleinberg1999fairness,kumar2000fairness}, of which approximate majorization is the most general \cite{goel2001approximate}. The notion of approximate majorization has been used in the metric distortion setting to study the fairness properties of social choice rules in the form of the fairness ratio \cite{goel2017metric}. Copeland is known to achieve a fairness ratio of 5 \cite{goel2017metric}, and a bound on the fairness ratio translates to a bound on social cost objectives belonging to a general class of convex functions \cite{goel2006simultaneous,goel2017metric}.

\section{Preliminaries}
\subsection{Social Choice Rules} Let $\V$ be the set of agents and $\cc$ the set of alternatives. We will use $N$ to denote the total number of agents, i.e., $N = |\V|$, and $m$ the number of alternatives, i.e., $m = |\cc|$. Every agent $v \in \V$ has a strict (no ties) preference ordering $\sigma_v$ on $\cc$. For any $c,\cpr \in \cc$, we will use $c \succ_v \cpr$ to denote the fact that agent $v \in \V$ \emph{prefers} $c$ over $\cpr$ in her ordering $\sigma_v$. Let $\s$ be the set of all possible preference orderings on $\cc$. We call a profile of preference orderings $\sigma \in \s^N$ an \emph{instance}. 

Based on the preferences of the agents, we either want to determine a single alternative as the winner, or a set of alternatives of a given size $\ell$.  A (deterministic) single-winner social choice rule is a function $f: \s^N \to \cc$ that maps each instance to an alternative. An $\ell$-winner social choice rule is a function $f: \s^N \to \{S \subseteq \cc : |S|= \ell\}$ that maps each instance to a set of $\ell$ alternatives.

To define the social choice rules that we use in this paper, we need one additional definition. We say that an alternative $c$ \emph{pairwise beats} $\cpr$ if $|\{v \in \V : c \succ_v \cpr\}| \geq \frac{N}{2}$, with ties broken arbitrarily. 

\begin{itemize}
\item \textbf{Copeland}: Given an instance $\sigma$, define a score for each alternative $c \in \cc$ given by $|\{c^\prime \in \cc : c \mbox{ pairwise beats } \cpr\}|$. The alternative with the highest score (the largest number of pairwise victories, with ties broken arbitrarily) is chosen as the winner.
\item \textbf{Recursive rules}: For any single-winner rule $f$, we define the $\ell$-winner ``Recursive $f$'' (e.g., Recursive Copeland). Given an instance $\sigma$, choose $\ell$ winners as follows: First pick the winner $c_1$ under rule $f$ among all alternatives $\cc$, then pick the winner $c_2$ under $f$ among $\cc \setminus \{c_1\}$, and so on until $c_\ell$. The set of winners is given by $\{c_1, c_2, \ldots, c_\ell\}$.
\item \textbf{STV}: Given an instance $\sigma$, repeatedly eliminate the alternative ranked first by the fewest agents and remove this alternative from every ranking. The last remaining alternative is the winner.
\end{itemize}

\subsection{Metric costs} We assume that the agent costs over the alternatives are given by an underlying metric $d$ on $\cc \cup \V$, such that $d(v,c)$ is the cost of an agent $v$ for an alternative $c$.
\begin{definition} \label{def:metric}
A function $d : \cc \cup \V \times \cc \cup \V \to \Rplus$ is a metric if and only if $\forall x,y,z \in \cc \cup \V$, we have the following:
\begin{enumerate}
\item $d(x,y) \geq 0$
\item $d(x,x) = 0$
\item $d(x,y) = d(y,x)$
\item $d(x,z) \leq d(x,y) + d(y,z)$\label{eqn:tri-ineq}
\end{enumerate}
\end{definition}

We can do with a much simpler yet equivalent assumption on the agents' costs. We need to first define a q-metric (``q'' for quadrilateral) by replacing the triangle inequalities by ``quadrilateral'' inequalities (Condition \ref{eqn:quad-ineq} in the definition below).
\begin{definition} \label{def:qmetric}
A function $d : \V \times \cc \to \Rplus$ is a q-metric if and only if $\forall v,\vpr \in \V$, and $\forall c, \cpr \in \cc$, we have the following:
\begin{enumerate}
\item $d(v,c) \geq 0$
\item $d(v,c) \leq d(v,\cpr) + d(\vpr,\cpr) + d(\vpr,c)$ \label{eqn:quad-ineq}
\end{enumerate}
\end{definition}
The following equivalence has been shown in earlier work \cite{goel2017metric}.
\begin{lemma}[Goel et al. \cite{goel2017metric}]\label{thm:tri-eq-quad}
If $d$ is a q-metric, then there exists a metric $\dpr$ such that $d(v,c) = \dpr(v,c)$ for all $v \in \V$ and $c \in \cc$.
\end{lemma}
Henceforth, we will use the terms \emph{metric} and \emph{q-metric} interchangeably.

\subsection{Distortion}
We say that a metric $d$ is \emph{consistent} with an instance $\sigma$, if whenever any agent $v$ prefers $c$ over $\cpr$, then her cost for $c$ must be at most her cost for $\cpr$, i.e., $c \succ_v \cpr \implies d(v,c) \leq d(v,\cpr)$. We denote by $\rho(\sigma)$ the set of all metrics $d$ that are consistent with $\sigma$.

The social cost, $\phi$, of an alternative is defined as the sum of the agent costs for it. For any metric $d$, and any alternative $c \in \cc$, we define $\phi(c,d) = \sum_{v \in \V} d(v,c)$. The social cost of a set of alternatives is the sum of social costs of the constituent alternatives. For any set of alternatives $S \subseteq \cc$, we define
\[\phi(S,d) = \sum_{v \in \V} \sum_{c \in S} d(v,c).\]
Below, we will define distortion in terms of choosing sets of alternatives. The corresponding definitions for single-winner social choice rules are obtained by using singleton sets instead.

We view social choice rules as trying to approximate the optimal set of alternatives of a given size $\ell$, with knowledge of only the rankings $\sigma$, but not the underlying metric cost $d$ that induces $\sigma$. To measure how close a social choice rule gets to the optimal set of size $\ell$ in terms of social cost, we define the \emph{distortion} of any given set $S$ of size $\ell$ to be the ratio of the social cost of $S$ to the cost of the optimal set according to $d$:
\begin{align*}
\dist(S,d) = \frac{\phi(S,d)}{\min_{T \subseteq \cc: |T| = \ell} \phi(T,d)}.
\end{align*}
The worst-case distortion of a set alternatives $S \subseteq \cc$ for a given instance $\sigma$ is given by
\begin{align*}
\dist(S,\sigma) = \sup_{d \in \rho(\sigma)} \dist(S,d).
\end{align*}
The distortion of a social choice rule $f$ is said to be at most $\beta$, if for any instance $\sigma$, we have $\dist(f(\sigma),\sigma) \leq \beta$. In other words, the worst-case distortion of the set of alternatives chosen by $f$, over all possible instances $\sigma$, and all metrics $d \in \rho(\sigma)$ that are consistent with it, is at most $\beta$. 

\subsection{Fairness} \label{subsec:prelims-fairness}
Given an underlying metric, based on the alternative chosen, the costs incurred might vary widely among the agents. We want to formally quantify how ``fair'' choosing a particular alternative is. For this purpose, we look at social cost defined as the sum of $k$ largest agent costs, for all $1 \leq k \leq N$. For any metric $d$ and $S \subseteq \cc$, we define
\begin{align*}
\phi_k(S,d) = \max_{V \subseteq \V : |V| = k} \sum_{v \in V} \sum_{c \in S} d(v,c).
\end{align*}

We measure fairness by a worst-case bound on how well a social choice rule approximates (simultaneously for all $k$) the optimal set of alternatives in terms of the social cost given by $\phi_k$, with knowledge of only the rankings $\sigma$, but not the underlying metric $d$. To this end, we define the \emph{fairness ratio} of a given set of alternatives $S$ of size $\ell$ as
\begin{align*}
\fratio(S,d) =  \max_{1 \leq k \leq N} \frac{\phi_k(S,d)}{\min_{T \subseteq \cc : |T| = \ell} \phi_k(T,d)}.
\end{align*}
The worst-case fairness ratio of a set of alternatives $S$ of size $\ell$, for a given instance $\sigma$ then becomes
\begin{align*}
\fratio(S,\sigma) = \sup_{d \in \rho(\sigma)} \fratio(S,d).
\end{align*}
Furthermore, the fairness ratio of a social choice rule $f$ is said to be at most $\beta$, if for any instance $\sigma$, we have $\fratio(f(\sigma),\sigma) \leq \beta$.

Another reason for studying the fairness ratio is that for deterministic social choice rules, a bound on the fairness ratio translates to an approximation result with respect to a broad class of convex cost functions (see Goel et al. \cite{goel2017metric})

\section{Relating Distortion and Fairness}
In this section, we establish our main result, which ties the distortion and fairness ratio closely together (within an additive factor of 2). This result additionally implies bounds on the fairness ratio for several popular social choice rules for which only bounds on the distortion were previously known, including $k$-Approval, Veto, Plurality, Borda, and STV. Until Sect.~\ref{sec:sets}, we consider only single-winner social choice rules $f:\s^N \to \cc$.

We will prove that the distortion-fairness gap is at most 2 for any fixed metric, which in turn implies the gap is at most 2 for any instance (taking the supremum over all metrics) \emph{and} for any rule (taking the supremum over all instances). Perhaps surprisingly, although this relationship is established on the level of a specific metric, we will show that it is tight on the level of instances, i.e., there exist instances for which the distortion-fairness gap of a given alternative approaches 2.

First, we note the following trivial relationship between the distortion and fairness ratio.
\begin{theorem}\label{thm:leq}
For any instance $\sigma$ and alternative $c$,
\[\dist(c,\sigma) \leq \fratio(c,\sigma).\]
\end{theorem}
Theorem~\ref{thm:leq} also implies that if the fairness ratio of a single-winner rule $f$ is at most $\beta$, then the distortion of $f$ is at most $\beta$.

Next, we establish a corresponding lower bound on distortion, giving the desired gap of 2. Again, we establish the bound for fixed instances by proving it for individual metrics, and this also implies that the same bound applies to any single-winner rule $f$.
\begin{theorem}\label{thm:2}
For any instance $\sigma$ and any alternative $c$,
\[\fratio(c,\sigma) - 2 < \dist(c,\sigma).\]
\end{theorem}
Theorem~\ref{thm:2} also implies that if the distortion of a single-winner rule $f$ is at most $\beta$, then the fairness ratio of $f$ is at most $\beta + 2$.

Additionally, the bound of Theorem~\ref{thm:leq} is tight, in the sense that there exists a series of instances and an alternative $c$ for which the distortion-fairness gap approaches 2 as the number of agents approaches $\infty$. One such series of instances is illustrated in Fig.~\ref{fig:tight}.

\begin{figure}
\begin{center}
 \begin{tikzpicture}
	\tikzstyle{dot} = [draw, fill, circle, inner sep=1.5pt];
	
	\node (21) at (0,0) [dot] {};
	\node at (21) [below=5pt] {$v$};
	\node (2) at (2,0) [dot] {};
	\node at (2) [below=5pt] {$c_2$};
	\node (12) at (4,0) [dot] {};
	\node at (12) [below=5pt] {$\V \setminus v$};
	\node (1) at (6,0) [dot] {};
	\node at (1) [below=5pt] {$c_1$};
	
	\draw (21) -- (2) node [midway,above] {$1$};
	\draw (2) -- (12) node [midway,above] {$1$};
	\draw (12) -- (1) node [midway,above] {$1 - \delta$};
\end{tikzpicture}
\end{center}
\caption{A graph metric inducing a series of instances $\sigma_N$ for any $N = | \V |$, where $c_2 \succ_v c_1$ and $c_1 \succ_{v'} c_2$ for all $v' \neq v$. This metric shows that $\fratio(c_1,\sigma_N) \geq 3$ (letting $\delta \to 0$); however $\dist(c_1,\sigma_N) \to 1$ as $N \to \infty$.}\label{fig:tight}
\end{figure}
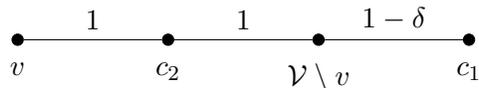

\begin{theorem}\label{thm:tight}
For the series of instances $\sigma_N$ illustrated in Fig.~\ref{fig:tight}, $\fratio(c_1,\sigma_N) - \dist(c_1,\sigma_N) \to 2$ as $N \to \infty$.
\end{theorem}

This example also implies that choosing the alternative with lowest distortion (in this case, $c_1$) can result in a distortion-fairness gap approaching $2$. Indeed, any reasonable social choice rule should choose $c_1$, which beats $c_2$ by an overwhelming majority, and thus would have a distortion-fairness gap approaching 2 for this instance.

In fact, this gap makes sense when we note that $\fratio(c,\sigma)$ is always at least 3 unless $c$ is ranked first by \emph{every} agent. Otherwise, let $v$ be an agent that prefers some alternative $c_\mathrm{adv}$ above $c$, i.e., $c_\mathrm{adv} \succ_v c$, and consider a metric $d$ such that
\begin{equation}\label{eq:d}
d(v',c') = \begin{cases}
3 &\text{if } v'=v, c_\mathrm{adv} \succ_v c' \\
1 &\text{otherwise}
\end{cases}.
\end{equation}
It is not hard to see that this metric satisfies the quadrilateral inequality and is consistent with $\sigma$. For this metric and $k=1$, we have
\[\frac{ \phi_1(c,d) }{ \min_{c' \in \cc} \phi_1(c',d) } = \frac{3}{1} = 3.\]
Since this consistent metric achieves a fairness ratio of 3, we conclude that $\fratio(c,\sigma) \geq 3$ unless every agent ranks $c$ first. Since distortion can approach 1 without this strict requirement, a gap of 2 naturally arises. (However, note that this flooring effect is not the only reason for a distortion-fairness gap; such a gap may still exist when the fairness ratio is larger than 3.)

\subsection{New Bounds on the Fairness of Specific Rules}
Now that we have established a close relationship between the distortion and fairness ratio -- within a constant additive factor -- we immediately get new results for any single-winner rule where distortion is known but fairness is not. These include many rules in common use: both simple rules such as $k$-Approval, Veto, Plurality, and Borda \cite{anshelevich2015approximating}; and more complex ones such as STV and the general class of scoring rules, i.e., rules where each position on an agent's ballot is worth a certain number of points and the alternative with the most points is the winner \cite{skowron2017social}. The following corollary arises directly from these known bounds and Theorems~\ref{thm:leq} and \ref{thm:2} establishing the relationship between the distortion and fairness ratio.

\begin{corollary}\label{cor:rules}
For an instance $\sigma$ with $N$ agents and $m$ alternatives,
\begin{enumerate}
\item The fairness ratio of $k$-Approval and Veto is $\Theta(N)$.
\item The fairness ratio of Plurality and Borda is $\Theta(m)$.
\item The fairness ratio of the harmonic scoring rule \footnote{see Skowron and Elkind \cite{skowron2017social} for a definition.} is $O(\frac{m}{\sqrt{\log m}})$ and $\Omega(\frac{m}{\log m})$.
\item The fairness ratio of \emph{any} scoring rule is $\Omega(\sqrt{\log m})$.
\item The fairness ratio of STV is $O(\log m)$ and $\Omega(\sqrt{\log m})$.
\end{enumerate}
\end{corollary}

\subsection{Calculating Fairness}
For completeness, we can give a straightforward way to calculate the fairness ratio exactly, similar to the program given by Goel et al. \cite{goel2017metric} for calculating distortion. Although we do not have a polynomial time algorithm for calculating the fairness ratio, it can be computed by means of a binary linear program. We provide details in Section \ref{app:calc_fairness} in the Appendix.

\section{Distortion and Fairness for Multiple-Winner Rules}\label{sec:sets}
The notions of distortion and fairness ratio extend naturally to choosing a \emph{set} of winners rather than a single alternative. In this section, we show that the same distortion bounds for single-winner rules can also be obtained for multiple-winner rules (under the sum of costs objective), and thus we can use Recursive Copeland to achieve a distortion of 5 even for multiple-winner rules. Additionally, we note that the close relationship between the distortion and fairness ratio extends to multiple-winner rules, and thus Recursive Copeland has a fairness ratio of at most 7.

First, we demonstrate a simple way to turn a single-winner social choice rule $f$ into an $\ell$-winner rule $f'$ for any $\ell$ and which has the \emph{same} distortion bound.
\begin{theorem}\label{thm:distsets}
Let $f$ be a single-winner social choice rule having distortion at most $\beta$. Let $\ell > 1$. Then the $\ell$-winner social choice rule $f' =$ ``Recursive $f$'' has distortion at most $\beta$.
\end{theorem}

Thus any distortion results obtainable for single-winner social choice rules can also be obtained for $\ell$-winner rules when we care about the sum of distances from an agent to each chosen alternative. Unfortunately, the results for fairness cannot be easily extended in the same way. In particular, the last few steps of the preceding proof relied on the ability to decompose the numerator and denominator in the calculation of distortion by separating out the costs for the individual alternatives. That is, we used the fact that
\[\phi(c_1,d) + \phi(c_2,d) = \phi(\{c_1,c_2\},d).\]
The analogous statement for fairness does \emph{not} hold, i.e.,
\[\phi_k(c_1,d) + \phi_k(c_2,d) \neq \phi_k(\{c_1,c_2\},d),\]
because the specific set of agents used to calculate the social cost $\phi_k$ can differ in these three cases unless $k=N$. Thus the left hand side of this equation could conceivably be larger than the right hand side, e.g., if the sets of agents farthest from $c_1$ and those farthest from $c_2$ have little overlap. Intuitively, this suggests that the iterative rule $f'$ might perform worse than the original rule $f$ in terms of fairness ratio, because even if each chosen alternative individually has a low fairness ratio, combining the adversarial alternatives into one set could decrease the denominator.

Nevertheless, it turns out that the fairness ratio of the iterative social choice rule $f'$ can't be too much worse than the fairness ratio of $f$, because the relationship between the distortion and fairness ratio extends to multiple-winner rules.
\begin{theorem}\label{thm:gapsets}
For any instance $\sigma$ any set of alternatives $S$,
\[\fratio(S,\sigma) - 2 < \dist(S,\sigma) \leq \fratio(S,\sigma).\]
\end{theorem}

\begin{corollary}
Let $f$ be a single-winner social choice rule which has fairness ratio at most $\beta$. Let $\ell > 1$. Then $\ell$-winner social choice rule $f' = $ ``Recursive $f$'' has fairness ratio at most $\beta + 2$.
\end{corollary}

This result leads to a low constant-factor bound for the fairness ratio of Recursive Copeland (though we conjecture that the true fairness ratio is lower).
\begin{corollary}\label{cor:recursive}
The $\ell$-winner rule Recursive Copeland has distortion at most 5 and fairness ratio at most 7.
\end{corollary}

\section{Conclusion}
We demonstrated that the distinct notions of distortion and fairness ratio are in fact closely linked -- within an additive factor of 2. This further lead to new bounds on the fairness ratio for several common social choice rules for which only bounds on the distortion were previously known, including STV and various scoring rules. Additionally, we showed that the distortion of any single-winner rule can also be obtained by a recursive multiple-winner rule. Together with the relationship between the distortion and fairness ratio for multiple-winner rules, this implied that Recursive Copeland achieves distortion at most 5 and fairness ratio at most 7.

\section*{Acknowledgments}
Ashish Goel and Anilesh K. Krishnaswamy were supported in part by NSF grant no. CCF-1637418,
ONR grant no. N00014-15- 1-2786 and ARO grant no. W911NF-14-1-0526. Reyna Hulett's research was
supported in part by NSF GRFP grant DGE-1656518.

\appendix

\section{All proofs}
\begin{proof}[Proof of Theorem \ref{thm:leq}]
Observe that for any fixed metric $d$, $\phi_N(c,d) = \phi(c,d)$, and therefore
\begin{align*}
\dist(c,d) &= \frac{\phi(c,d)}{\min_{c' \in \cc} \phi(c',d)} \leq \max_{1 \leq k \leq N} \frac{\phi_k(c,d)}{\min_{c' \in \cc} \phi_k(c',d)} \\
&= \fratio(c,d) \leq \fratio(c,\sigma),
\end{align*}
since $\fratio(c,\sigma)$ is the worst-case fairness ratio over all metrics. This implies $\fratio(c,\sigma)$ upper bounds $\dist(c,d)$ for every $d$, and thus by the definition of supremum,
\[\dist(c,\sigma) = \sup_{d \in \rho(\sigma)} \dist(c,d) \leq \fratio(c,\sigma),\]
as desired.
\qed\end{proof}

\begin{proof}[Proof of Theorem \ref{thm:2}]
Denote the $N$ agents of $\sigma$ by $\V$ and the $m$ alternatives by $\cc$. We will establish that for any fixed $k$ and metric $d$,
\begin{equation}\label{eq:goal}
\frac{\phi_k(c,d)}{\min_{c' \in \cc} \phi_k(c',d)} - 2\frac{N-1}{N} \leq \dist(c,d),
\end{equation}
and thus taking the maximum over all $k$,
\begin{align*}
\fratio(c,d) - 2\frac{N-1}{N} &= \max_{1 \leq k \leq N} \frac{\phi_k(c,d)}{\min_{c' \in \cc} \phi_k(c',d)} - 2\frac{N-1}{N} \\
& \leq \dist(c,d) \leq \dist(c,\sigma),
\end{align*}
since $\dist(c,\sigma)$ is the worst-case distortion over all metrics. Thus by the definition of supremum,
\[\fratio(c,\sigma) - 2 = \sup_{d \in \rho(\sigma)} \fratio(c,d) - 2 < \dist(c,\sigma),\]
as desired.

For convenience, we denote the adversarial alternative for fairness by $\opt = \argmin_{c' \in \cc} \phi_k(c',d)$, and the set of $k$ agents farthest from $c$ by $V_k = \argmax_{V \subseteq \V : |V| = k} \sum_{v \in V} d(v, c)$. Furthermore, let the numerator and denominator of the fairness ratio be denoted respectively as $f_k = \phi_k(c,d)$ and $g_k = \phi_k(\opt,d)$. Thus we can rewrite the desired result \ref{eq:goal} as
\[\frac{f_k}{g_k} - 2\frac{N-1}{N} = \frac{\sum_{v \in V_k} d(v,c)}{ \phi_k(\opt,d) } - 2\frac{N-1}{N} \leq \frac{ \phi(c,d) }{ \min_{c' \in \cc} \phi(c',d) }.\]

We proceed by calculating a lower bound on $\frac{ \phi(c,d) }{ \phi(\opt,d) }$, which will immediately also apply to $\frac{ \phi(c,d) }{ \min_{c' \in \cc} \phi(c',d) }$. We divide the agents into two groups, $V_k$ and $\V \setminus V_k$. Note that by definition, $\sum_{v \in V_k} d(v, c) = f_k$, and furthermore that 
\begin{align*}
\sum_{v \in V_k} d(v, \opt) \leq \max_{V \subseteq \V : |V| = k} \sum_{v \in V} d(v, \opt) = g_k, 
\end{align*}
as shown in Fig.~\ref{fig:graph}.

\begin{figure}
\input{dist_graph}
\caption{A partial representation of the metric $d$ highlighting the set of $k$ agents farthest from $c$, $V_k$, and an arbitrary agent $v \notin V_k$, with distances. Note that the top two distances refer to \emph{total} distance to $V_k$.}\label{fig:graph}
\end{figure}
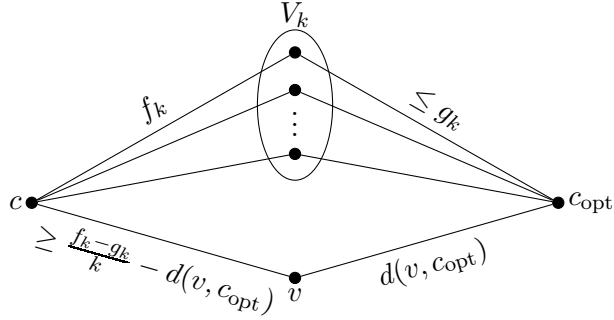

Now, consider any $v \notin V_k$. For any $w \in V_k$, we have $d(w,c) \leq d(w,\opt) + d(v,c) + d(v,\opt)$ by the quadrilateral inequality. Solving for $d(v,c)$ and averaging over all $k$ nodes in $V_k$, we get
\begin{align}\label{eq:z}
d(v,c) &\geq \frac{\sum_{w \in V_k} \left( d(w,c) - d(w,\opt) - d(v,\opt) \right)}{k} \nonumber \\ 
&\geq \frac{f_k - g_k}{k} - d(v, \opt).
\end{align}

We can now calculate a lower bound on the distortion of $c$ relative to $\opt$.
\begingroup
\allowdisplaybreaks
\begin{align*}
\frac{\phi(c,d)}{\phi(\opt,d)} &= \frac{ f_k + \sum_{v \notin V_k} d(v,c) }{ \sum_{v \in \V} d(v, \opt) } \\
&\geq \frac{ f_k + \sum_{v \notin V_k} d(v,c) }{ N \frac{g_k}{k} } \stepcounter{equation}\tag{\theequation}\label{eq:avg} \\
&\geq \frac{ f_k + (N-k)\frac{f_k-g_k}{k} - \sum_{v \notin V_k} d(v,\opt) }{ N \frac{g_k}{k} } \stepcounter{equation}\tag{\theequation}\label{eq:zz} \\
&= \frac{f_k}{g_k} - \frac{N-k}{N} - \frac{\sum_{v \notin V_k} d(v, \opt)}{N \frac{g_k}{k}},
\end{align*}%
\endgroup
where (\ref{eq:avg}) follows because the average distance from an agent to $\opt$ can't be more than $\frac{1}{k}$ times the sum of the largest $k$ distances, and (\ref{eq:zz}) follows from (\ref{eq:z}). We now consider two cases. In the first, let $k \leq \frac{N}{2}$. Then $|\V \setminus V_k| \geq k$ so we can repeat the argument that the average distance from an agent in $\V \setminus V_k$ can't be more than $\frac{1}{k}$ times the sum of the largest $k$ distances for any agents. Thus
\begin{align*}
\frac{ \phi(c,d) }{ \phi(\opt,d) } &\geq \frac{f_k}{g_k} - \frac{N-k}{N} - \frac{\sum_{v \notin V_k} d(v, \opt)}{N \frac{g_k}{k}} \\
&\geq \frac{f_k}{g_k} - \frac{N-k}{N} - \frac{(N-k)\frac{g_k}{k}}{N \frac{g_k}{k}} \\
&= \frac{f_k}{g_k} - 2\frac{N-k}{N} \geq \frac{f_k}{g_k} - 2\frac{N-1}{N},
\end{align*}
as desired. (Note that the last inequality is tight when $k=1$.) On the other hand, if $k \geq \frac{N}{2}$, then $|\V \setminus V_k| \leq k$ so we can only say $\sum_{v \notin V_k} d(v, \opt) \leq g_k$. Thus
\begin{align*}
\frac{ \phi(c,d) }{ \phi(\opt,d) } &\geq \frac{f_k}{g_k} - \frac{N-k}{N} - \frac{\sum_{v \notin V_k} d(v, \opt)}{N \frac{g_k}{k}} \\
&\geq \frac{f_k}{g_k} - \frac{N-k}{N} - \frac{g_k}{N \frac{g_k}{k}} \\
&= \frac{f_k}{g_k} - 1 \geq \frac{f_k}{g_k} - 2\frac{N-1}{N},
\end{align*}
as desired, where the last inequality holds provided $N \geq 2$. (When $N=1$, distortion and fairness are identical so the desired result holds regardless.)
\qed\end{proof}

\begin{proof}[Proof of Theorem \ref{thm:tight}]
As shown in Fig.~\ref{fig:tight}, for $\sigma_N$, we let $\cc = \{c_1,c_2\}$ where one agent $v$ votes $c_2 \succ_v c_1$ and the remaining $N-1$ agents vote $c_1 \succ_{v'} c_2$ for $v' \in \V \setminus \{v\}$. As $N \to \infty$, the distortion of $c_1$ approaches 1, but as the example metric shows, the fairness ratio of $c_1$ is 3. Thus $\fratio(c_1,\sigma_N) - \dist(c_1,\sigma_N)$ approaches 2 as $N$ approaches $\infty$.
\qed\end{proof}

\begin{proof}[Proof of Theorem \ref{thm:distsets}]
We wish to show that for any instance $\sigma$ with $N$ agents $\V$ and $m \geq \ell$ alternatives $\cc$, $\dist(f',\sigma) \leq \beta$. Equivalently, for any fixed metric $d$ and adversarial set $T$ of size $\ell$, we wish to show
\[ \frac{ \phi(f'(\sigma),d) }{ \phi(T,d) } \leq \beta. \]
Let $f'(\sigma) = \{c_1, \dots, c_\ell\}$ in the order in which they are selected by Recursive $f$. Furthermore, order $T$ as $\{c_1', \dots, c_\ell'\}$ such that the elements $f'(\sigma) \cap T$ (if any) appear first and in the same order as in $f'(\sigma)$. Note that this implies if any $c_j = c_i' \in T$, then $i \leq j$. Let $\sigma_i = \sigma \setminus \{c_1, \dots, c_{i-1}\}$ be the instance obtained from $\sigma$ by deleting candidates $\{c_1,\dots,c_{i-1}\}$ from every ranking, so $c_i = f(\sigma_i)$, and analogously let $\cc_i$ be the set of alternatives remaining in $\sigma_i$, $\cc_i = \cc \setminus \{c_1,\dots,c_{i-1}\}$. Then we know that
\[ \sup_{d \in \rho(\sigma_i)} \frac{\phi(f(\sigma_i),d)}{\min_{c \in \cc_i} \phi(c,d)}  \leq \beta, \]
simply because the distortion of $f$ is at most $\beta$.

However, observe that the specific metric $d$ we are considering satisfies $d \in \rho(\sigma) \subseteq \rho(\sigma_i)$ (because any consistent metric for $\sigma$ is also consistent when considering only the alternatives $\cc_i$ remaining in $\sigma_i$). Additionally, we must have $c_i' \in \cc_i = \cc \setminus \{c_1,\dots,c_{i-1}\}$ because of the way we ordered the alternatives of $T$ -- otherwise, we would have $c_i' = c_j$ for some $j < i$, which is a contradiction. Thus, we can look specifically at the metric $d$ and adversarial alternative $c_i'$ to get
\[ \frac{\phi(f(\sigma_i),d)}{\phi(c_i',d)}  \leq \beta. \]
Since this holds for every $i$, we can combine these equations to obtain
\begin{align*}
\frac{\phi(f(\sigma_1),d) + \cdots + \phi(f(\sigma_\ell),d)}{\phi(c_1',d) + \cdots + \phi(c_\ell',d)} &\leq \beta \\
\frac{\sum_{c \in f'(\sigma)} \phi(c,d)}{\sum_{c \in T} \phi(c,d)} &\leq \beta \\
\frac{\phi(f'(\sigma),d)}{\phi(T,d)} &\leq \beta,
\end{align*}
as desired.
\qed\end{proof}

\begin{proof}[Proof of Theorem \ref{thm:gapsets}]
As with single-winner rules, $$\dist(S,\sigma) \leq \fratio(S,\sigma)$$ holds because fairness bounds the largest $k$ agent costs simultaneously for all $k$, including $k=N$ which is equivalent to distortion.

To show $\fratio(S,\sigma) - 2 < \dist(S,\sigma)$, we first observe that the quadrilateral inequality holds for fixed-size \emph{sets} of alternatives. That is, if we define $d(v,S) = \sum_{c \in S} d(v,c)$ then for any $v,w \in \V$, $S,T \subseteq \cc$, $|S|=|T|=\ell$, we have $d(v,S) \leq d(v,T) + d(w,S) + d(w,T)$. This can be established from the quadrilateral inequality $d(v,c_1) \leq d(v,c_2) + d(w,c_1) + d(w,c_2)$ applied to every pair $c_1 \in S,c_2 \in T$. Summing these inequalities over the $\ell^2$ such pairs and dividing by $\ell$ gives the desired statement for sets of $\ell$ elements. The remainder of the proof is identical to that of Theorem~\ref{thm:2} where we replace the single alternatives $c,\opt$ with sets $S,S_\mathrm{opt}$, so we omit the details.
\qed\end{proof}

\begin{proof}[Proof of Corollary \ref{cor:recursive}]
By Theorem~\ref{thm:leq}, if the fairness ratio of $f$ is at most $\beta$ then the distortion of $f$ is also at most $\beta$. Theorem~\ref{thm:distsets} then implies the distortion of $f'$ is at most $\beta$. Thus using Theorem~\ref{thm:gapsets}, for any instance $\sigma$, we have
\[\fratio(f'(\sigma),\sigma) < \dist(f'(\sigma),\sigma) + 2 \leq \beta + 2.\]
By definition, this means the fairness ratio of $f'$ is at most $\beta+2$, as desired.
\qed\end{proof}

\section{Calculating Fairness}\label{app:calc_fairness}
In what follows, we outline a straightforward way to calculate the fairness ratio exactly, similar to the program given by Goel et al. \cite{goel2017metric} for calculating distortion. Although we do not have a polynomial time algorithm for calculating the fairness ratio, it can be computed with the following binary linear program on an instance $\sigma$ and alternative $c$ for a fixed $k$ and adversarial alternative $\opt$. (Note that as with the program of Goel et al. \cite{goel2017metric}, this program could be used on any specific instance to choose the alternative with lowest worst-case fairness ratio -- though in this case the resulting social choice rule would not necessarily run in polynomial time.) Here we use $\{v_i : i \in [N]\}$ to represent the agents.

\begin{algorithm}
\begin{equation*}
\begin{array}{rrclcl}
\max & \multicolumn{3}{l}{\displaystyle \sum_{i=1}^N d_i} \\
\\
\textrm{s.t.} & d_i & \leq & d(v_i, c) & & \forall i \\
& d_i & \leq & M b_i & & \forall i \\
&\displaystyle \sum_{i=1}^N b_i & \leq & k \\
\\
&\text{binary } b_i \\
\\
& kt + \displaystyle \sum_{i=1}^N d_i^{(\mathrm{opt})} & \leq & 1 \\
& d_i^{(\mathrm{opt})} & \geq & 0 & & \forall i \\
& d_i^{(\mathrm{opt})} & \geq & d(v_i,\opt) - t \\
\\
& d & \in & \rho(\sigma)
\end{array}
\end{equation*}
\caption{A binary linear program for calculating fairness.}\label{alg:blp}
\end{algorithm}

In the above program, $M$ is any sufficiently large number, e.g., $3^m$ where $m$ is the number of alternatives. Recall that $\rho(\sigma)$ is the set of all valid metrics consistent with $\sigma$; this requirement can also be encoded with a polynomial number of inequalities on the distances $d(v_i,c_j)$ for $i \in [N], j \in [m]$ \cite{goel2017metric}.

We will say a word about the choice of $M$. The program above essentially works by allowing the $d_i$'s to represent the $k$ largest distances $d(v_i, c)$, and 0 for any other distances. (Meanwhile, the $k$ largest distances $d(v_i,\opt)$ are represented by $d_i^{(\mathrm{opt})}$, normalized so that their sum is at most 1; thus maximizing $\sum_{i=1}^N d_i$ is equivalent to maximizing the worst-case fairness \emph{ratio}.) Thus when $b_i =0$, we have $d_i \leq M b_i = 0$, so $v_i$ is not included in the largest $k$ distances, and when $b_i = 1$ we want $d_i \leq M b_i = M$ to be unrestrictive. That is, we would like $M$ to be sufficiently large that, subject to the normalization of the distances to $\opt$, $d(v_i,c) \leq M$ for any $i$ and metric $d$.

Strictly speaking, no $M$ can guarantee this, because for some instances $\sigma$ the fairness ratio can be unbounded. However, the following lemma gives a simple way to tell whether the fairness ratio of one alternative over another is unbounded, which will also allow us to give an upper bound on $M$ when it \emph{is} bounded.

\begin{lemma}
Consider an instance $\sigma$ with $|\cc| = m$ alternatives $c_i$ and $|\V| = N$ agents $v_i$. Furthermore, consider a directed graph on $\cc$, $D=(\cc,E)$, where $(c_i,c_j) \in E$ whenever at least one agent prefers $c_i$ over $c_j$, i.e., whenever $\exists v c_i \succ_v c_j$. Then the fairness ratio (and therefore the distortion) of an alternative $\alg$ over $\opt$ is bounded if and only if there exists a directed path $\alg \to \opt$ in $D$. 
\end{lemma}
\begin{proof}
($\implies$) We prove the contrapositive, namely, that if no directed path $\alg \to \opt$ exists, then the fairness ratio is unbounded. Let $S$ be the set of vertices (alternatives) which have a directed path to $\opt$; by assumption $\alg \notin S$. Consider the following metric: $S$ and $\V$ are co-located at one point, and the remaining alternatives $\cc \setminus S$, are co-located at another; the distance between the two points is 1. (Alternatively, we can say the distance between any pair of ``co-located'' points is $\epsilon \ll 1$.) It is not hard to see that this satisfies the triangle inequality. Furthermore, it is consistent. This holds because every agent must rank every alternative in $S$ above every alternative in $\cc \setminus S$; otherwise there would exist $v, c \in \cc \setminus S, c' \in S$ such that $c \succ_v c'$, implying $(c,c') \in E$. But this, along with the existence of a directed path $c' \to \opt$, implies the existence of a path $c \to \opt$, which contradicts $c \notin S$. Thus we have defined a consistent metric, and in this metric, the fairness ratio is $\frac{1}{\epsilon}$ which approaches $\infty$ as $\epsilon \to 0$.

($\impliedby$) Consider an arbitrary edge $(c_i,c_j) \in E$, and let $d_j$ be any upper bound $d_j \geq d(v,c_j) \ \forall v$. Since $(c_i,c_j) \in E$, there must exist some agent $v'$ such that $c_i \succ_{v'} c_j$. Thus using the quadrilateral inequality, for any agent $v$, we have
\[d(v,c_i) \leq d(v,c_j) + d(v',c_i) + d(v',c_j) \leq d(v,c_j) + 2d(v',c_j) \leq 3 d_j \]
Thus inductively, if there exists a path $\alg \to \opt$, which must have length at most $|\cc| = m$, then $d(v,\alg) \leq 3^{m-1} d_\mathrm{opt}$ where $d_\mathrm{opt}$ is any upper bound on $d(v,\opt)$. For instance, we can set $d_\mathrm{opt} = \max_{v\in \V} d(v, \opt)$. Then the fairness ratio of $\alg$ over $\opt$ for any metric is at most
\[\max_{1 \leq k \leq N} \frac{\phi_k(\alg,\sigma)}{\phi_k(\opt,\sigma)} \leq \frac{N 3^{m-1} d_\mathrm{opt}}{d_\mathrm{opt}} = N 3^{m-1}, \]
and thus is bounded.
\qed\end{proof}

The above proof incidentally suggests a way of setting $M$ to be unrestrictive when the distortion and fairness ratio are bounded. Specifically, since our binary linear program normalizes \[\max_{V \subseteq \V : |V| = k} \sum_V d(v,\opt) \leq 1,\] we can upper bound $d(v,\opt)$ with $d_\mathrm{opt} = 1$ which implies $d(v,c) \leq 3^{m-1} d_\mathrm{opt} = 3^{m-1}$ if the fairness ratio is bounded. Thus we can set $M = 3^m$, and we will get the correct answer whenever the fairness ratio is bounded (and the output will be $k 3^m$ if and only if the fairness ratio is unbounded).

\end{document}

%% file: dist_graph.tex
\begin{center}
 \begin{tikzpicture}
	\tikzstyle{dot} = [draw, fill, circle, inner sep=1.5pt]
	
	\node (alg) at (-.5,2) [dot] {};
	\node at (alg) [left] {$c$};
	\node (opt) at (6.5,2) [dot] {};
	\node at (opt) [right] {$\opt$};
	
	\node (v) at (3,1) [dot] {};
	\node at (v) [below] {$v$};
	
	\node (s1) at (3,2.65) [dot] {};
	\node (s2) at (3,3.5) [dot] {};
	\node (s3) at (3,4) [dot] {};
	\node at (3,3.15) {$\vdots$};
	\node at (3,4.5) {$V_k$};
	\draw (3,3.3) ellipse (0.5cm and 1cm);
	
    \draw (alg) -- (s1);
    \draw (alg) -- (s2);
    \draw (alg) -- (s3) node [midway,sloped,above] {$f_k$};
    \draw (alg) -- (v) node [midway,sloped,below] {$\geq \frac{f_k-g_k}{k} - d(v,\opt)$};
    
    \draw (opt) -- (s1);
    \draw (opt) -- (s2);
    \draw (opt) -- (s3) node [midway,sloped,above] {$\leq g_k$};
    \draw (opt) -- (v) node [midway,sloped,below] {$d(v,\opt)$};

\end{tikzpicture}
\end{center}

%% file: fairness_netecon.bbl
\begin{thebibliography}{10}

\bibitem{andersen2001resilient}
David Andersen, Hari Balakrishnan, Frans Kaashoek, and Robert Morris.
\newblock {\em Resilient overlay networks}, volume~35.
\newblock ACM, 2001.

\bibitem{anshelevich2015approximating}
Elliot Anshelevich, Onkar Bhardwaj, and John Postl.
\newblock Approximating optimal social choice under metric preferences.
\newblock {\em Association for the Advancement of Artificial Intelligence, 15th
  Conference of the}, 2015.

\bibitem{anshelevich2016randomized}
Elliot Anshelevich and John Postl.
\newblock Randomized social choice functions under metric preferences.
\newblock {\em 25th International Joint Conference on Artificial Intelligence},
  2016.

\bibitem{arrow2012social}
Kenneth~J Arrow.
\newblock {\em Social choice and individual values}, volume~12.
\newblock Yale university press, 2012.

\bibitem{arya2004local}
Vijay Arya, Naveen Garg, Rohit Khandekar, Adam Meyerson, Kamesh Munagala, and
  Vinayaka Pandit.
\newblock Local search heuristics for k-median and facility location problems.
\newblock {\em SIAM Journal on computing}, 2004.

\bibitem{barbera2001introduction}
Salvador Barbera.
\newblock An introduction to strategy-proof social choice functions.
\newblock {\em Social Choice and Welfare}, 2001.

\bibitem{barbera2008choose}
Salvador Barber{\`a} and Danilo Coelho.
\newblock How to choose a non-controversial list with k names.
\newblock {\em Social Choice and Welfare}, 31(1):79--96, 2008.

\bibitem{betzler2013computation}
Nadja Betzler, Arkadii Slinko, and Johannes Uhlmann.
\newblock On the computation of fully proportional representation.
\newblock {\em Journal of Artificial Intelligence Research}, 47:475--519, 2013.

\bibitem{boutilier2015optimal}
Craig Boutilier, Ioannis Caragiannis, Simi Haber, Tyler Lu, Ariel~D Procaccia,
  and Or~Sheffet.
\newblock Optimal social choice functions: A utilitarian view.
\newblock {\em Artificial Intelligence}, 2015.

\bibitem{cabannes2004participatory}
Yves Cabannes.
\newblock Participatory budgeting: a significant contribution to participatory
  democracy.
\newblock {\em Environment and urbanization}, 16(1):27--46, 2004.

\bibitem{caragiannis2017subset}
Ioannis Caragiannis, Swaprava Nath, Ariel~D Procaccia, and Nisarg Shah.
\newblock Subset selection via implicit utilitarian voting.
\newblock {\em Journal of Artificial Intelligence Research}, 2017.

\bibitem{caragiannis2011voting}
Ioannis Caragiannis and Ariel~D Procaccia.
\newblock Voting almost maximizes social welfare despite limited communication.
\newblock {\em Artificial Intelligence}, 2011.

\bibitem{chamberlin1983representative}
John~R Chamberlin and Paul~N Courant.
\newblock Representative deliberations and representative decisions:
  Proportional representation and the {B}orda rule.
\newblock {\em American Political Science Review}, 1983.

\bibitem{drezner1995facility}
Zvi Drezner and Horst~W Hamacher.
\newblock {\em Facility location}.
\newblock Springer-Verlag New York, NY, 1995.

\bibitem{elkind2017properties}
Edith Elkind, Piotr Faliszewski, Piotr Skowron, and Arkadii Slinko.
\newblock Properties of multiwinner voting rules.
\newblock {\em Social Choice and Welfare}, 48(3):599--632, 2017.

\bibitem{enelow1984spatial}
James~M Enelow and Melvin~J Hinich.
\newblock {\em The spatial theory of voting: An introduction}.
\newblock CUP Archive, 1984.

\bibitem{faliszewski2017multiwinner}
Piotr Faliszewski, Piotr Skowron, Arkadii Slinko, and Nimrod Talmon.
\newblock Multiwinner voting: A new challenge for social choice theory.
\newblock {\em Trends in Computational Social Choice}, page~27, 2017.

\bibitem{feldman2016voting}
Michal Feldman, Amos Fiat, and Iddan Golomb.
\newblock On voting and facility location.
\newblock {\em Proceedings of the 2016 ACM Conference on Economics and
  Computation}, 2016.

\bibitem{filos2014truthful}
Aris Filos-Ratsikas and Peter~Bro Miltersen.
\newblock Truthful approximations to range voting.
\newblock {\em International Conference on Web and Internet Economics}, 2014.

\bibitem{gibbard1973manipulation}
Allan Gibbard.
\newblock Manipulation of voting schemes: a general result.
\newblock {\em Econometrica: journal of the Econometric Society}, 1973.

\bibitem{goel2017metric}
Ashish Goel, Anilesh~K Krishnaswamy, and Kamesh Munagala.
\newblock Metric distortion of social choice rules: Lower bounds and fairness
  properties.
\newblock In {\em Proceedings of the 2017 ACM Conference on Economics and
  Computation}, pages 287--304. ACM, 2017.

\bibitem{goel2015knapsack}
Ashish Goel, Anilesh~K Krishnaswamy, Sukolsak Sakshuwong, and Tanja Aitamurto.
\newblock Knapsack voting.
\newblock {\em Collective Intelligence}, 2015.

\bibitem{goel2006simultaneous}
Ashish Goel and Adam Meyerson.
\newblock Simultaneous optimization via approximate majorization for concave
  profits or convex costs.
\newblock {\em Algorithmica}, 2006.

\bibitem{goel2001approximate}
Ashish Goel, Adam Meyerson, and Serge Plotkin.
\newblock Approximate majorization and fair online load balancing.
\newblock {\em Symposium on Discrete Algorithms: Proceedings of the twelfth
  annual ACM-SIAM symposium on Discrete algorithms}, 2001.

\bibitem{kleinberg1999fairness}
Jon Kleinberg, Yuval Rabani, and {\'E}va Tardos.
\newblock Fairness in routing and load balancing.
\newblock {\em Foundations of Computer Science, 1999. 40th Annual Symposium
  on}, 1999.

\bibitem{kumar2000fairness}
Amit Kumar and Jon Kleinberg.
\newblock Fairness measures for resource allocation.
\newblock {\em Foundations of Computer Science, 41st Annual Symposium on},
  2000.

\bibitem{lu2011budgeted}
Tyler Lu and Craig Boutilier.
\newblock Budgeted social choice: From consensus to personalized decision
  making.
\newblock In {\em IJCAI}, volume~11, pages 280--286, 2011.

\bibitem{monroe1995fully}
Burt~L Monroe.
\newblock Fully proportional representation.
\newblock {\em American Political Science Review}, 1995.

\bibitem{moulin1980strategy}
Herv{\'e} Moulin.
\newblock On strategy-proofness and single peakedness.
\newblock {\em Public Choice}, 1980.

\bibitem{moulin2016handbook}
Herv{\'e} Moulin, Felix Brandt, Vincent Conitzer, Ulle Endriss, Ariel~D
  Procaccia, and J{\'e}r{\^o}me Lang.
\newblock {\em Handbook of Computational Social Choice}.
\newblock Cambridge University Press, 2016.

\bibitem{procaccia2006distortion}
Ariel~D Procaccia and Jeffrey~S Rosenschein.
\newblock The distortion of cardinal preferences in voting.
\newblock {\em International Workshop on Cooperative Information Agents}, 2006.

\bibitem{satterthwaite1975strategy}
Mark~Allen Satterthwaite.
\newblock Strategy-proofness and arrow's conditions: Existence and
  correspondence theorems for voting procedures and social welfare functions.
\newblock {\em Journal of economic theory}, 1975.

\bibitem{skowron2017social}
Piotr~Krzysztof Skowron and Edith Elkind.
\newblock Social choice under metric preferences: Scoring rules and {STV}.
\newblock {\em AAAI}, 2017.

\end{thebibliography}
